\documentclass[runningheads]{llncs}
\usepackage[T1]{fontenc}
\usepackage{amsmath,amsfonts}

\usepackage{algorithmic}
\usepackage{algorithm}
\usepackage{array}
\usepackage[caption=false,font=normalsize,labelfont=sf,textfont=sf]{subfig}
\usepackage{textcomp}
\usepackage{stfloats}
\usepackage{url}
\usepackage{verbatim}
\usepackage{graphicx}
\usepackage{cite}
\usepackage{caption}
\usepackage{subcaption}

\usepackage{listings}
\lstdefinelanguage{JavaScript}{
  keywords={break, case, catch, continue, debugger, default, delete, do, else, finally, for, function, if, in, instanceof, new, return, switch, this, throw, try, typeof, var, void, while, with},
  morecomment=[l]{//},
  morecomment=[s]{/*}{*/},
  morestring=[b]',
  morestring=[b]",
  sensitive=true
}
\usepackage{xcolor}
\definecolor{codegreen}{rgb}{0,0.6,0}
\definecolor{codegray}{rgb}{0.5,0.5,0.5}
\definecolor{codepurple}{rgb}{0.58,0,0.82}
\definecolor{backcolour}{rgb}{0.95,0.95,0.92}

\lstdefinestyle{mystyle}{
    backgroundcolor=\color{backcolour},   
    commentstyle=\color{codegreen},
    keywordstyle=\color{magenta},
    numberstyle=\tiny\color{codegray},
    stringstyle=\color{codepurple},
    basicstyle=\ttfamily\footnotesize,
    breakatwhitespace=false,         
    breaklines=true,                 
    captionpos=b,                    
    keepspaces=true,                 
    numbers=left,                    
    numbersep=5pt,                  
    showspaces=false,                
    showstringspaces=false,
    showtabs=false,                  
    tabsize=2
}

\colorlet{punct}{red!60!black}
\definecolor{background}{HTML}{EEEEEE}
\definecolor{delim}{RGB}{20,105,176}
\colorlet{numb}{magenta!60!black}

\newtheorem{assumption}{Assumption}

\def\AA{\mathcal{A}}
\def\BB{\mathcal{B}}
\def\CC{\mathcal{C}}

\def\PP{\mathcal{P}}
\def\Addr{\texttt{Addr}}
\def\Sum{\texttt{Sum}}
\def\Simplify{\texttt{Simplify}}

\def\Bal{\texttt{Bal}}

\lstdefinelanguage{json}{
    numbers=left,
    numberstyle=\scriptsize,
    stepnumber=1,
    numbersep=8pt,
    breaklines=true,
    backgroundcolor=\color{background},
    literate=
     *{0}{{{\color{numb}0}}}{1}
      {1}{{{\color{numb}1}}}{1}
      {2}{{{\color{numb}2}}}{1}
      {3}{{{\color{numb}3}}}{1}
      {4}{{{\color{numb}4}}}{1}
      {5}{{{\color{numb}5}}}{1}
      {6}{{{\color{numb}6}}}{1}
      {7}{{{\color{numb}7}}}{1}
      {8}{{{\color{numb}8}}}{1}
      {9}{{{\color{numb}9}}}{1}
      {:}{{{\color{punct}{:}}}}{1}
      {,}{{{\color{punct}{,}}}}{1}
      {\{}{{{\color{delim}{\{}}}}{1}
      {\}}{{{\color{delim}{\}}}}}{1}
      {[}{{{\color{delim}{[}}}}{1}
      {]}{{{\color{delim}{]}}}}{1},
}

\newcolumntype{C}{>{\centering\arraybackslash}X} 
\newcolumntype{L}{>{\raggedright\arraybackslash}X} 

\begin{document}

\title{Shedding Light on Complex Bitcoin Mixer Transactions: 67-Fold Reduction in Unclassified Cases}

\author{Nikolay Larionov\inst{1} \and
Yekaterina Smolenkova\inst{2} \and
Yury Yanovich\inst{2}}
\titlerunning{Shedding Light on Complex Bitcoin Mixer Transactions}
\authorrunning{N. Larionov et al.}
%

\institute{Moscow Institute of Physics and Technology, Moscow, Russia  \and
Skolkovo Institute of Science and Technology, Moscow, Russia}

\maketitle

\begin{abstract}
Bitcoin's Unspent Transaction Output (UTXO) model enables public analysis of fund flows, but users often merge transactions into Shared Send Mixers (SSMs) to obscure these flows. Untangling SSMs to recover original subtransactions is an NP-complete problem. While a practical untangling algorithm exists, it fails to classify 1.4\% of SSM transactions due to computational time limits. This paper introduces four novel heuristics that exploit structural weaknesses in real-world SSM transactions to resolve these timeout cases: preemptive grouping, connectable singleton, ambiguous pairing, and knapsack fallback. We provide theoretical proofs validating each heuristic and integrate them into an optimized pipeline. Applied to timeout transactions, our approach classifies 98.5\% of previously unresolved cases, reducing the overall unclassified transaction rate from 1.4\% to 0.021\% of all SSM transactions. Our open-source implementation and comprehensive evaluation on the complete Bitcoin blockchain demonstrate that the heuristics effectively untangle previously intractable transactions, enabling more accurate flow analysis and deeper structural insights into cryptocurrency transaction patterns.

\end{abstract}

\section{Introduction}
\label{sec:introduction}

Bitcoin stands as the pioneering and largest cryptocurrency by market capitalization built on blockchain technology~\cite{Nakamoto2008, CoinMarketCap2026}. Its globally distributed network employs a Proof-of-Work consensus mechanism, enabling censorship-resistant cross-border transfers of its native digital asset--bitcoin \cite{BedfordTaylor2017,Ren2019,Tovanich2021,Badea2021}. The protocol does not mandate user identity disclosure; instead, it requires only cryptographic proof of ownership: public key hashes when receiving funds and digitally signed transactions with corresponding private keys when sending. Bitcoin encodes recipient information as addresses, which include format identifiers and checksums for validation \cite{Wuille2017,Wu2021}. While data encryption methods exist, they often introduce complexity, higher operational costs, and reduced trust \cite{Ben-Sasson2014,Noether2016,Kugusheva2019a,Korepanova2019a}. Consequently, user anonymity primarily relies on pseudonymity--interacting solely via platform-specific identifiers (addresses) without revealing real-world identities \cite{Liu2021,VoundiKoe2022}. This censorship resistance prohibits transaction freezing or user bans, while the immutable ledger ensures confirmed transactions cannot be reversed. Combined with pseudonymity, these attributes have made Bitcoin and similar cryptocurrencies attractive to criminal organizations for money laundering, illicit commerce, fraud, and ransomware operations \cite{Europol2021a}.

Bitcoin utilizes an Unspent Transaction Output (UTXO) model for currency circulation \cite{Nakamoto2008}. Transactions comprise multiple inputs and outputs, where each output specifies a non-negative value and spending conditions--typically requiring a digital signature from the private key matching the output's embedded public key hash. The public transaction history enables analysis of fund flows, but users often merge multiple payments into single transactions--Shared Send Mixers (SSMs)--to obscure these flows \cite{Maxwell2013}. Despite their longevity, SSMs remain effective privacy mechanisms, as corroborated by recent Europol operational analyses \cite{Europol2021a}.

Kristov Atlas pioneered SSM untangling in 2014 through a Sudoku-inspired matching framework \cite{Atlas2014}. Subsequent research formally established the problem's computational hardness (proving it NP-complete) and demonstrated that exact solutions are constrained to pseudo-polynomial time algorithms~\cite{Yanovich2016b}. Building on this work, a practical algorithm was proposed and applied to the complete Bitcoin blockchain, showing that 15\% of transactions are SSMs, with 90\% uniquely untanglable \cite{Larionov2023}. However, this implementation leaves 1.4\% of SSM transactions unresolved due to computational time limits, creating a significant gap in transaction analysis capabilities.

This study advances SSM untangling by introducing a suite of heuristics specifically designed to resolve timeout transactions. Our contributions are:

\begin{itemize}
    \item \textbf{Four novel heuristics} for SSM untangling: (1) preemptive grouping, (2) connectable singleton, (3) ambiguous pairing, and (4) knapsack fallback--each with theoretical correctness proofs.
    \item \textbf{An optimized pipeline} that applies heuristics in optimal order to maximize classification efficiency, reducing the unclassified transaction rate from 1.4\% to 0.021\% of all SSM transactions--a 98.5\% improvement.
    \item \textbf{Open-source implementation} of the heuristic pipeline, enabling reproducible research and practical deployment.
    \item \textbf{Comprehensive empirical evaluation} on historical Bitcoin data (up to block 882,421), providing detailed statistics on heuristic effectiveness and resulting transaction type distributions.
\end{itemize}

The remainder of this paper is organized as follows: Section~\ref{sec:related_work} surveys related work; Section~\ref{sec:untangling} provides the definition of untangling problem; Section~\ref{sec::Heu} presents new heuristics and provides theoretical proofs of their correctness; Section~\ref{sec::pipepline} details the experimental setup of heuristic application; Section~\ref{sec::Experiments} describes the results of numerical experiments; and Section~\ref{sec:conclusion} concludes with future work.

\section{Related Work}
\label{sec:related_work}

Privacy preservation in digital systems resembles a continuous contest where one party seeks to conceal information while another aims to expose it. In the context of Bitcoin, the privacy model is fundamentally built upon maintaining the anonymity of public keys that serve as transaction endpoints~\cite{Nakamoto2008}. Since the publication of the original Bitcoin whitepaper, various Unspent Transaction Output (UTXO) templates have been developed and deployed, including Pay to Public Key (P2PK), Pay to Public Key Hash (P2PKH), Pay to Script Hash (P2SH), Pay to Witness Public Key Hash (P2WPKH), and Pay to Taproot (P2TR)~\cite{Lombrozo2015}. These different templates utilize public keys in distinct ways, and due to the cryptographic hashing operations involved, it is often not possible to compute one template from another directly~\cite{Menezes1996}. Consequently, researchers in the field typically consider an address as a commitment to a public key rather than the public key itself, which has important implications for how addresses can be linked and analyzed.

Users of the Bitcoin network are free to generate as many addresses as they desire without any practical limitations. To use a generated address for spending funds, the user must retain and securely store the corresponding private key or keys that authorize transactions from that address. The introduction of Hierarchical Deterministic (HD) wallets revolutionized key management by enabling users to derive multiple addresses from a single seed key, making the process of managing numerous addresses as simple as maintaining a single private key~\cite{Pieter2012}. This development, while convenient for users, inadvertently created a significant analytical challenge for researchers attempting to group addresses by their controlling entities or wallets. This challenge has motivated extensive research into heuristics that can infer ownership relationships from on-chain transaction patterns.

Bitcoin usage patterns have inspired several heuristics designed to reveal information about address ownership and transaction flows~\cite{Androulaki2013, Reid2013}. The common spending heuristic operates on the observation that when multiple UTXOs are used as inputs to fund a single transaction, those inputs are likely controlled by the same entity, as users typically combine multiple unspent outputs to achieve a desired payment amount~\cite{Meiklejohn2013}. The one-time change heuristic identifies change addresses based on the pattern that occurs when the total input amount exceeds the intended payment amount, resulting in the user receiving change back to an address they control~\cite{Meiklejohn2013, Moser2022}. The coinbase clustering heuristic groups outputs of miner reward transactions, based on the observation that mining pools consolidate their earnings regardless of the number of individual UTXOs generated~\cite{Eck2021}. While these heuristics have proven valuable for address clustering and transaction analysis, they remain fundamentally probabilistic in nature and are therefore prone to errors and misclassifications in certain scenarios~\cite{Larionov2024}.

One significant source of errors in heuristic application is the presence of shared send mixers, which intentionally obfuscate transaction lineage by combining multiple payments from different users into a single transaction~\cite{Maxwell2013}. These mixers simultaneously tangle the transaction history and reduce the effective fee per transfer by splitting fixed transaction costs among multiple participants. The features offered by shared send mixers made them a popular component of custodial wallet implementations before approximately 2017, and recent reports from Europol indicate that these mechanisms continue to be exploited in illicit operations for money laundering and other criminal activities~\cite{Europol2021a, IOCTA2021}. The computational foundations for untangling shared send mixer transactions were formally established through mathematical analysis that proved the problem to be NP-complete, demonstrating that exact solutions require computational effort that grows exponentially with transaction complexity~\cite{Yanovich2016b}. Despite this theoretical hardness, a practical untangling algorithm was later developed and implemented, providing the first comprehensive analysis of mixer transactions across the entire Bitcoin blockchain history~\cite{Larionov2023}.

The practical implementation of the untangling algorithm revealed that approximately fifteen percent of all Bitcoin transactions qualify as shared send mixers according to the formal definition requiring multiple inputs and multiple outputs~\cite{Larionov2023}. Among these mixer transactions, about ninety percent could be uniquely untangled, meaning that a single minimal partition of inputs and outputs could be identified. However, approximately 1.4\% of mixer transactions remained unclassified because they exceeded a predefined computational time limit of one second per transaction. Manual inspection of these timeout transactions revealed recurrent structural patterns that, while computationally challenging for generic subset-sum search algorithms, exhibited exploitable characteristics that could potentially be leveraged to either resolve the transactions outright or dramatically reduce the search space required for analysis.

The timeout transactions that resist classification are frequently characterized by having large numbers of inputs and outputs combined with intricate value relationships that make them resemble challenging instances of the knapsack problem. Indeed, the core requirement for untangling--that for any candidate grouping, the total value of selected inputs must fall between the total value of their corresponding outputs and that same sum plus the transaction fee--creates a constraint mathematically equivalent to checking whether a subset of input amounts can achieve a target sum defined by the outputs. This is precisely a variant of the subset-sum problem, establishing a direct connection to the well-studied knapsack problem. The knapsack problem has received extensive attention in the operations research literature, with researchers developing both exact algorithms such as branch-and-bound methods and dynamic programming approaches, as well as heuristic algorithms including greedy methods and scaling techniques that trade optimality for computational efficiency~\cite{Martello1987}. More recently, machine learning approaches have been applied to learn heuristics for knapsack instances by training neural networks on input-output examples, demonstrating that learned representations can capture useful structural information about item relationships and capacity constraints~\cite{Nomer2020}. While these machine learning methods show promise for certain problem classes, they have not yet been integrated into practical blockchain analytics pipelines for transaction analysis.

Beyond the core task of transaction untangling, researchers have developed numerous complementary de-anonymization techniques that enhance address clustering and entity identification. Off-chain data sources provide publicly available tags and labels that can be linked to on-chain activity, with services like WalletExplorer offering categorized organization labels and Bitcoin Abuse maintaining databases of scam reports associated with specific addresses~\cite{Chainalysis2023, BitcoinAbuse2023}. These external data sources enable delayed or real-time correlation of blockchain activity with real-world entities, supporting compliance and investigative efforts. Network-level attacks have also been proposed for de-anonymization purposes, including Sybil attacks and fake node deployments that aim to harvest IP addresses of transaction originators~\cite{Biryukov2015}. In response, countermeasures such as transaction remote release and delayed propagation have been developed to protect user privacy at the network layer~\cite{ShenTu2015b, Fanti2018}. Machine learning models increasingly integrate on-chain features with off-chain information to perform sophisticated analytical tasks, including clustering addresses by wallet ownership~\cite{Ermilov2017, Zhang2020, Moser2022}, classifying wallets by entity type and behavior patterns~\cite{Toyoda2019, Liu2021, Nerurkar2022}, and computing risk scores for anti-money laundering compliance and regulatory reporting~\cite{Weber2019, Zhang2020a, Day2021, Wahrstatter2023}.

Recent research has expanded the scope of address linkage analysis by investigating public key reuse across different address formats within the Bitcoin network itself~\cite{Kalodner2020}. The KeyLinker method demonstrates that when a public key appears in any transaction input, all addresses derived from that same public key across different formats can be reliably linked, providing cryptographically grounded evidence for address grouping that does not rely on probabilistic heuristics~\cite{Smolenkova2025}. This approach leverages the practical impossibility of private key collisions and the collision resistance of the hash functions used in address generation to establish deterministic connections between seemingly unrelated addresses. Extending this concept further, subsequent work has examined public key reuse across multiple cryptocurrencies spanning both UTXO-based and account-based models, revealing that cryptographic keys are extensively and continually reused across networks such as Bitcoin, Ethereum, Litecoin, Dogecoin, Zcash, and Tron~\cite{Stutz2026}. These cross-chain key reuse patterns have significant implications for privacy and security, as they enable entity clustering across heterogeneous blockchain systems that would otherwise appear disconnected.

While these various lines of research provide valuable tools and insights for blockchain analysis, they address different aspects of the broader de-anonymization challenge. The specific problem of untangling shared send mixer transactions remains distinct from address clustering and key reuse analysis, focusing instead on recovering the internal flow of funds within individual transactions that deliberately combine multiple payments. The existing literature provides a solid theoretical and practical foundation for SSM analysis but leaves a specific and significant gap: the approximately one point four percent of mixer transactions that resist current untangling algorithms due to computational timeout limitations~\cite{Larionov2023}. These unresolved transactions represent a critical gap in transaction analysis capabilities, potentially obscuring fund flows that could be relevant for forensic investigations, regulatory compliance, and academic research.

The present paper addresses this gap directly by introducing four novel heuristics specifically designed to exploit structural weaknesses observed in real-world timeout transactions. These heuristics are grounded in theoretical proofs of correctness and are integrated into an optimized pipeline that applies them in an order designed to maximize classification efficiency while minimizing computational overhead. By targeting the structural patterns that cause transactions to exceed time limits in generic untangling algorithms, our approach resolves ninety-eight point five percent of previously timeout transactions, reducing the overall unclassified transaction rate from 1.4\% to approximately 0.021\% of all shared send mixer transactions. This sixty-seven-fold reduction in unclassified cases enables near-complete analysis of mixer transactions across the Bitcoin blockchain, opening the door to more accurate flow tracing, deeper structural insights into cryptocurrency transaction patterns, and enhanced capabilities for forensic investigations and regulatory compliance applications.

\section{Transactions untangling}
\label{sec:untangling}

In this section we formulate the untangling problem as it was stated in \cite{Larionov2023}. 
A Bitcoin address serves as a unique identifier for a UTXO payment instruction. While Bitcoin addresses can provide additional information in the scope of the current problem they are treated only as identifier.
In each transaction, an address and a non-negative amount are used to represent the input and output. To differentiate between them, we use lowercase letters for amounts, uppercase letters for addresses, and calligraphic capital letters for multisets of amount and address pairs. For instance, $(c_k, C_k)$ denotes an input or output with the address $C_k$ and the amount $\CC = \cup_{k=1}^K {(c_k, C_k)}$ represents a collection of such pairs.

Coinbase transactions are omitted since they do not involve inputs, are not shared send mixers, and do not appear after the untangling process.

A Bitcoin transaction is represented as an ordered triple $t = (\AA, \BB, c)$, where:
\begin{itemize}
\item $\AA$ is a multiset of transaction inputs, where each input $(a_n, A_n) \in \AA$ consists of an ordered pair with address $A_n$ and value $a_n \geq 0$.
\item $\BB$ is a multiset of transaction outputs, where each output $(b_m, B_m) \in \BB$ consists of an ordered pair with address $B_m$ and value $b_m \geq 0$.
\item $c = \sum\limits_{(a_n, \cdot) \in \AA} a_n - \sum\limits_{(b_m, \cdot) \in \BB} b_m \geq 0$ represents the transaction fee.
\end{itemize}
For any arbitrary multiset of transaction inputs or outputs $\CC$, $\Addr(\CC) = \cup_{(\cdot, C) \in \CC} {C}$ denotes the multiset of addresses in $\CC$, and $\Sum(\CC) = \sum_{(c, \cdot) \in \CC} c$.

An address may appear multiple times within the set of inputs and/or outputs. To streamline the untangling process, a $\Simplify$ function is used to consolidate multiple usages of the same address. $\Simplify$ processes the multisets of inputs and outputs in the following manner:
\begin{enumerate}
    \item
        The inputs are grouped by addresses and the sums of each group are calculated. The same is done for the outputs.
    \item
        For each address $C$ that exists both as an input and an output, the smaller amount between the input and output is substituted for both.
    \item
        Any pairs with zero amounts are discarded.
\end{enumerate}
I.e., 
$\Simplify: \; t \mapsto t'$,
where
\begin{align*}
    & t=(\AA, \BB, c), \; t'=(\AA', \BB', c'), \\
    & \AA' = \{(-\Bal(t, A), A) | A \in \Addr(\AA) \wedge \Bal(t, A) <0\}, \\
    & \BB' = \{(\Bal(t, B), B) | B \in \Addr(\BB) \wedge \Bal(t, B) >0\},  \\
    & c' = c,
\end{align*}
and $\Bal$ is the balance function
\begin{align*}
    \Bal(t, C) \equiv \sum\limits_{(b,C) \in \BB} b - \sum\limits_{(a,C) \in \AA} a.
\end{align*}

The untangling procedure works with simplified transactions. We omit any special notations for the simplified transactions that are possible. 

Let $|\cdot|$ be the cardinality operator. Denote $N \equiv |\AA|$ and $M \equiv |\BB|$.

Public transaction information includes details about inputs and outputs of the transactions, but not internal currency flows inside it. Shared send analysis aims to infer or reconstruct these missing edges based on the available vertex data and additional information. It is based on the following assumption:
\begin{assumption}
The intended expenditure from each subset of inputs does not exceed their actual expenditure:
\begin{align*}
        \forall \BB' \subset \BB \colon \; \sum\limits_{
        \begin{array}{l}
            (a, A) \in \AA\colon \; \exists (b, B) \in \BB'  \\ \wedge \left((a,A), (b, B)\right) \in E
        \end{array}} a \geq \sum\limits_{(b, B) \in \BB} b.
    \end{align*}
\end{assumption} 

\begin{definition}
\label{def:SST}
A (simplified) transaction $t=(\AA, \BB, c)$ is called a \textit{shared send transaction} if and only if it has multiple inputs and multiple outputs: $N\equiv|\AA|>1$ and $M\equiv|\BB|>1$. If $N\equiv|\AA|\leq 1$ and $M\equiv|\BB|=1$, the transaction is not a shared send, and we call it \textit{regular}.
\end{definition}

A regular transaction cannot be classified as a shared send because it involves exactly one input or output, meaning there is no possibility for funds to be split or shared among multiple inputs or outputs, respectively.

\begin{definition}
\label{def:connectable}
Given a transaction $t = (\AA, \BB, c)$, a \textit{pair} of sets $\AA' \subset \AA$ and $\BB' \subset \BB$, with at least one non-empty set, is called \textit{connectable} if and only if the following condition holds:
$$\Sum(\BB') + c \geq \Sum(\AA') \geq \Sum(\BB').$$
\end{definition}

Intuitively, connectivity means that the bitcoins from the addresses in $\AA'$ were spent exclusively to the addresses in $\BB'$, and similarly, the funds from $\AA \setminus \AA'$ were spent to $\BB \setminus \BB'$.

A collection of sets $\{X_k\}_{k=1}^K$ is called a partition of a set $X$ if the sets in the collection are pairwise disjoint, and their union equals $X$. We denote a partition as $X = \sqcup_{k=1}^K X_k$.

 \begin{definition}
     Consider a transaction $t = (\AA, \BB, c)$. A pair of $K$-element \textit{partitions} of input and
     output sets, i.e. $$\AA = \sqcup_{k=1}^K \AA_k, \;\; \BB = \sqcup_{k=1}^K \BB_k,$$ is called \textit{acceptable} iff each pair $(\AA_k, \BB_k)$ is connectable.
 \end{definition}

We denote such partitions as $$\PP = \{(\AA_1, \BB_1), \dots, (\AA_K, \BB_K)\},$$ where sets $\AA_k$ and $\BB_k$ are said to correspond to each other.

Since merging any two connectable pairs in a partition yields another valid partition, the total number of possible partitions can be extremely large. To render the problem tractable, untangling considers only minimal partitions, thereby reducing the solution space.

\begin{definition}
     A  \textit{connectable pair} $(\AA, \BB)$ is called \textit{minimal} iff the sets in it does not allow smaller partitions:
     $\forall \AA_1, \AA_2, \BB_1, \BB_2 \colon$ $\AA = \AA_1 \sqcup \AA_2$ and $\BB = \BB_1 \sqcup \BB_2$ $\rightarrow$ at least one pair $(\AA_1, \BB_1)$ or $(\AA_2, \BB_2)$ is not connectable.   
 \end{definition}

 \begin{definition}
 \label{def:minimalConnectable}
     An acceptable \textit{partition} $\PP = \{(\AA_k, \BB_k)\}_{k=1}^K$ is called \textit{minimal} iff it cannot be further subdivided into an acceptable partition, i.e., $(\AA_k, \BB_k)$ is a minimal connectable pair for each $k = 1, \dots, K$.   
 \end{definition}

Using the concept of minimal partitions into connectable pairs, we formally state the untangling problem.

\begin{definition}
    The minimal untangling problem is, based on the transaction $t = (\AA, \BB, c)$, to produce every minimal partition of the transaction graph. 
\end{definition}

This restriction implicitly assumes that a non-mixing transaction is unlikely to mimic the structure of a mixing transaction. Furthermore, value flows are more likely to correspond to a minimal acceptable partition than to a non-minimal one, as creating divisible flows would require additional, deliberate measures.

\begin{definition}
\label{def:txType}
Based on the solution to the minimal untangling problem, a shared send transaction $t=(\AA, \BB, c)$ is classified as:
\begin{itemize}
\item \textit{Simple} if and only if there are no non-trivial acceptable partitions.
\item \textit{Separable} if and only if there exists a unique minimal partition, and it contains disconnected components.
\item \textit{Ambiguous} if and only if there are at least two distinct minimal partitions.
\end{itemize}
\end{definition}

\subsection{Transaction examples}

Here we demonstrate examples of transactions of different classes (Figure \ref{figure:ex_combined}). Simple transactions are the transactions where none connectable pairs except from the non-trivial. For separable transactions where a unique minimal partition exists with disconnected components. This task requires to check all the possible subsets making the task similar to knapsack problem and making it NP-complete. The ambiguous (Figure \ref{figure:ex_ambiguous}) transactions represent the transaction where the search for the single untangling found several possible different untanglings. 

\begin{figure}[h!]
    \centering
    \begin{minipage}[b]{0.48\linewidth}
        \centering
        \includegraphics[width=\linewidth]{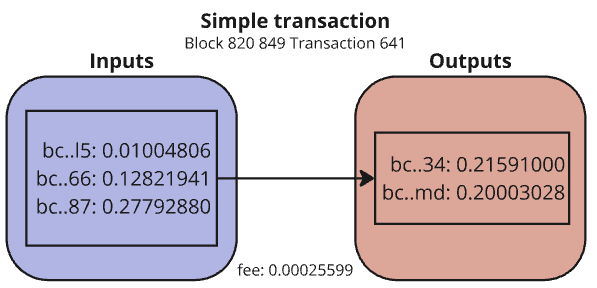}
        \vspace{0.5em}
        \textbf{a)} Example of simple transaction.
        \label{figure:ex_simple}
    \end{minipage}
    \hfill
    \begin{minipage}[b]{0.48\linewidth}
        \centering
        \includegraphics[width=\linewidth]{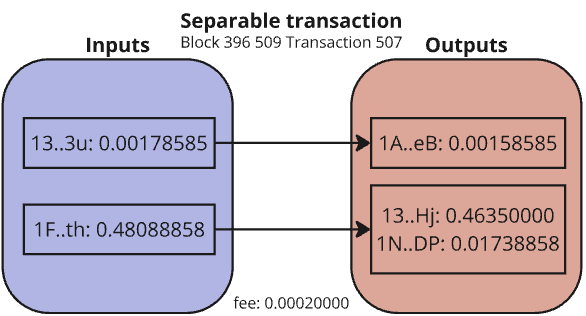}
        \vspace{0.5em}
        \textbf{b)} Example of separable transaction.
        \label{figure:ex_separable}
    \end{minipage}
    \caption{Transaction examples. \textbf{(a)} No subset of the inputs can form a connectable pair with only one output, so the only connectable pair is the whole transaction. \textbf{(b)} Input 13..3u forms a connectable pair with output 1A..eB, and the other input forms a pair with the other two.}
    \label{figure:ex_combined}
\end{figure}



\begin{figure}[h!]
    \centering
    \includegraphics[width=1\linewidth]{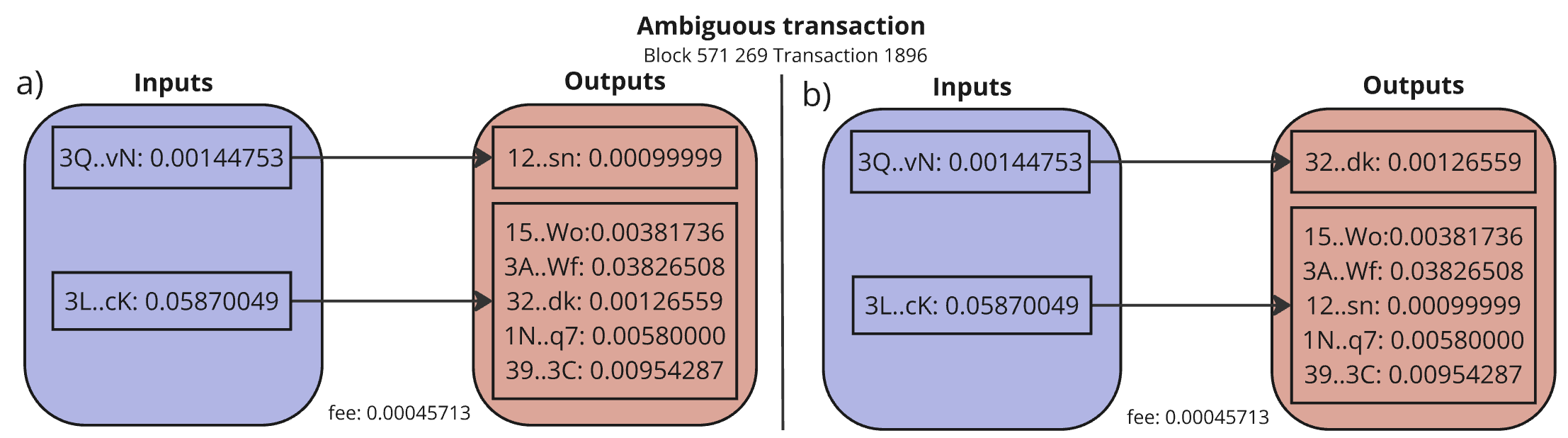}
    \caption{Examples of ambiguous transaction. In this transaction input 3Q..vN and output 12..sn form the connectable pair, and adresses 3Q..vN, 32..dk form a connectable pair. So transaction has different untanglings.}
    \label{figure:ex_ambiguous}
\end{figure}

\section{Untangling Heuristics}
\label{sec::Heu}

The NP-completeness of the untangling problem renders brute-force search infeasible for transactions with numerous inputs and outputs, leading to the ``time limit'' classification in practical solvers \cite{Larionov2023}. Manual inspection of such timeout transactions reveals recurrent structural patterns that, while computationally challenging for generic subset-sum search, can be exploited to either resolve the transaction outright or dramatically prune the search space. This section formalizes four such patterns as heuristics, each accompanied by a proof of correctness.

\subsection{Heuristic 1: Preemptive Grouping}

This heuristic identifies scenarios where a single input or output is so large that it must be paired with its counterpart, allowing the pair to be collapsed into a single entity.

\begin{theorem}
\label{thm:heuristic1}
Consider a simplified transaction $t = (\AA, \BB, c)$.
\begin{enumerate}
    \item If there exists an input $(a, A) \in \AA$ such that $a > \Sum(\AA \setminus \{a\})$ and for every output $(b, B) \in \BB$, $b > \Sum(\AA \setminus \{a\})$, then any connectable pair that contains $A$ must also contain $B$, and vice versa.
    \item If there exists an output $(b, B) \in \BB$ such that $b > \Sum(\BB \setminus \{b\}) + c$ and for every input $(a, A) \in \AA$, $a > \Sum(\BB \setminus \{b\}) + c$, then any connectable pair that contains $B$ must also contain $A$, and vice versa.
\end{enumerate}
\end{theorem}

\begin{proof}
We prove the first case; the second follows by symmetry. Assume, for contradiction, that a connectable pair $(\AA', \BB')$ exists with $B \in \BB'$ but $A \notin \AA'$. By the condition, $\Sum(\BB') \geq b > \Sum(\AA \setminus \{a\})$. Since $A \notin \AA'$, we have $\AA' \subseteq \AA \setminus \{a\}$, and thus $\Sum(\AA') \leq \Sum(\AA \setminus \{a\})$. This yields $\Sum(\BB') > \Sum(\AA')$, which violates the connectability condition $\Sum(\AA') \geq \Sum(\BB')$ from Definition~\ref{def:connectable}. A symmetric contradiction arises if $A \in \AA'$ and $B \notin \BB'$, as the complement of a connectable pair is also connectable.
\end{proof}

\begin{figure}[h!]
    \centering
    \includegraphics[width=0.5\linewidth]{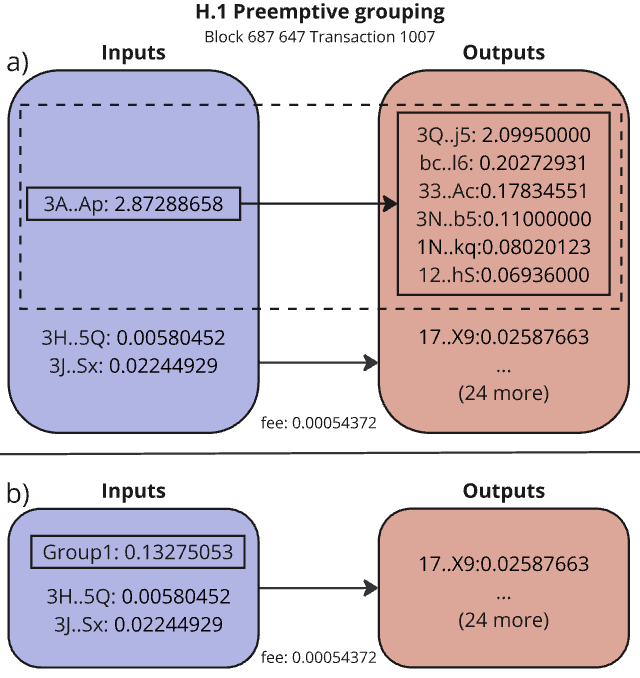}
    \caption{Application of Heuristic 1. Outputs 3Q..j5, bc..l6, 33..Ac, 3N..b5, 1N..kq, 12..hS each exceed the sum of all inputs except the largest (3A..Ap). Therefore, any connectable pair containing either output must also contain 3A..Ap, enabling consolidation into address Group1.}
    \label{figure:h1}
\end{figure}

When the conditions of Theorem~\ref{thm:heuristic1} are satisfied, the input $A$ and output $B$ are effectively locked together (Figure~\ref{figure:h1}). The heuristic merges them into a single logical entity: if $a \geq b$, the pair is replaced by a single input of value $a-b$ at address $A$; otherwise, it is replaced by a single output of value $b-a$ at address $B$. This reduction preserves all valid untanglings while strictly decreasing the cardinality of $\AA$ or $\BB$. The process can be applied iteratively until no further qualifying pairs exist.

\subsection{Heuristic 2: Connectable Singleton Reduction}

Transactions with exactly two inputs or exactly two outputs exhibit a structural symmetry that halves the search space.

\begin{theorem}
\label{thm:heuristic2}
For a transaction with $|\AA| = 2$, it is sufficient to consider only the smaller input when enumerating candidate connectable pairs. The same holds symmetrically for the smaller output when $|\BB| = 2$.
\end{theorem}

\begin{proof}
Let $\AA = \{A_1, A_2\}$ with $a_1 \leq a_2$. Any connectable pair $(\AA', \BB')$ either includes $A_1$ or it does not. If it does not include $A_1$, then $\AA' = \{A_2\}$ and its complement $\overline{\AA'} = \{A_1\}$ is also a valid partition component. Therefore, every valid partition can be reconstructed from the set of connectable pairs involving the smaller input $A_1$. The argument for outputs is identical.
\end{proof}

\begin{figure}[h!]
    \centering
    \begin{minipage}[b]{0.48\linewidth}
        \centering
        \includegraphics[width=\linewidth]{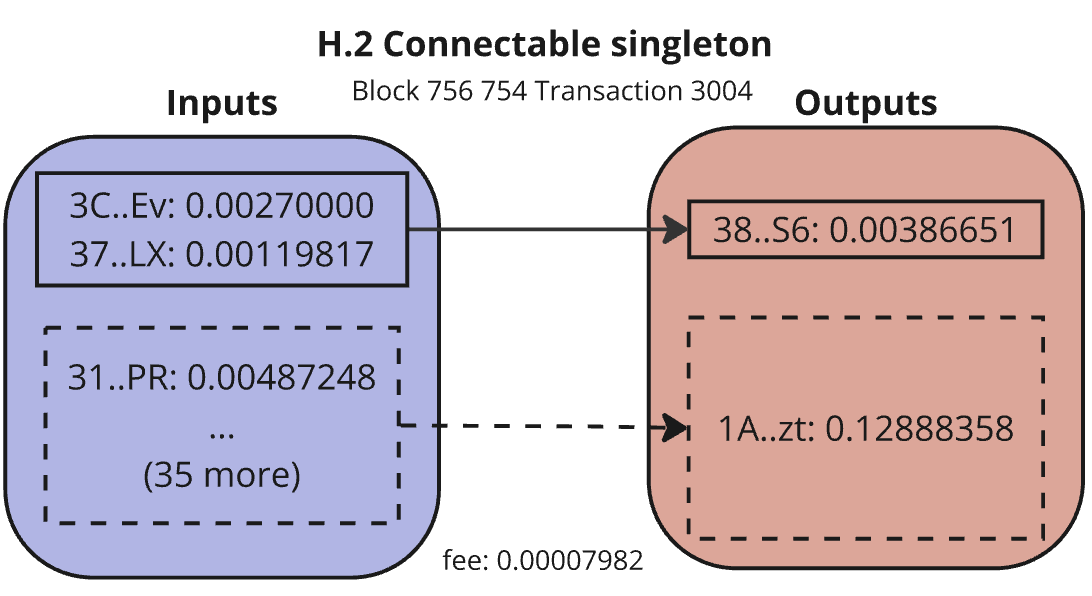}
        \vspace{0.5em}
        \textbf{a)} Application of Heuristic 2
        \label{figure:h2}
    \end{minipage}
    \hfill
    \begin{minipage}[b]{0.48\linewidth}
        \centering
        \includegraphics[width=\linewidth]{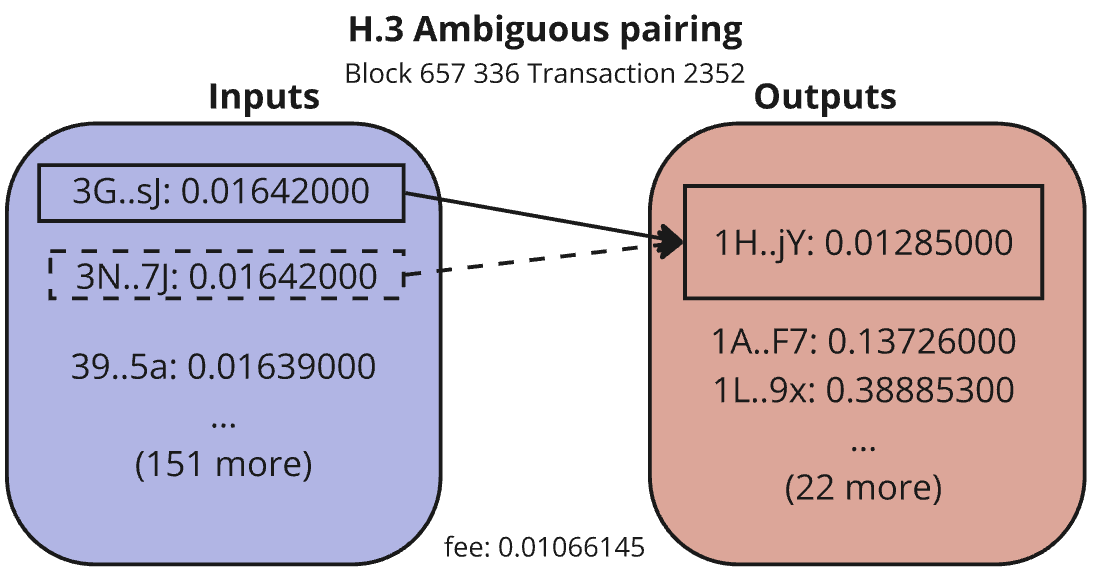}
        \vspace{0.5em}
        \textbf{b)} Application of Heuristic 3
        \label{figure:h3}
    \end{minipage}
    \caption{Application of heuristics. \textbf{(a)} Heuristic 2: With two inputs, examining only the smaller output (38..S6) is sufficient to identify the connectable pair (\{3C..Ev, 37..LX\}, 31..PR). Some other inputs (ex. 31..RP) can be safely ignored during this check as they exceed the value of 38..S6 + fee. \textbf{(b)} Heuristic 3: Inputs 3G..sJ and 3N..7J are of equal value. Since 3G..sJ forms a connectable pair with output 1H..jY, the existence of 3N..7J automatically classifies the transaction as ambiguous.}
    \label{figure:heuristics_combined}
\end{figure}


This theorem enables a targeted subset enumeration. Instead of iterating over all $2^{|\AA|}$ input subsets, the algorithm iterates only over subsets that include the smaller element and prunes the search once the cumulative sum exceeds $\min(\BB)+c$. For transactions with two outputs, the symmetric procedure is applied.

\subsection{Heuristic 3: Ambiguity via Identical Values}

Timeout transactions frequently contain multiple inputs or outputs of identical value. This seemingly minor artifact is a sufficient condition for ambiguity.

\begin{theorem}
\label{thm:heuristic3}
A transaction is \textit{ambiguous} if either of the following holds:
\begin{enumerate}
    \item There exist distinct inputs $a_i, a_j \in \AA$ with $a_i = a_j$ and a subset $\BB' \subseteq \BB$ such that $(\{a_i\}, \BB')$ is connectable.
    \item There exist distinct outputs $b_i, b_j \in \BB$ with $b_i = b_j$ and a subset $\AA' \subseteq \AA$ such that $(\AA', \{b_i\})$ is connectable.
\end{enumerate}
\end{theorem}

\begin{proof}
If $(\{a_i\}, \BB')$ is connectable and $a_i = a_j$, then by substituting the input, $(\{a_j\}, \BB')$ is also connectable. The existence of two distinct minimal partitions (one using $a_i$, the other using $a_j$) renders the transaction ambiguous by Definition~\ref{def:txType}. The proof for identical outputs is symmetric.
\end{proof}


This heuristic is computationally cheap: it requires only detecting value collisions among inputs and outputs that participate in singleton connectable pairs. Its high yield on timeout transactions stems from the prevalence of fixed-denomination UTXOs in mixer implementations.

\subsection{Heuristic 4: Knapsack Fallback and Ambiguity Propagation}

Despite optimizations, some transactions require solving a subset-sum problem. We employ a bounded knapsack solver to identify candidate connectable pairs. However, the solver's output is treated as evidence, not a definitive classification.

\begin{theorem}
\label{thm:heuristic4}
Let a connectable pair $(\AA', \BB')$ be discovered via knapsack search.
\begin{enumerate}
    \item If re-running the solver with a reduced fee $\hat{c} < c$ yields another distinct pair $(\AA'', \BB')$ for the same $\BB'$, the transaction is \textit{ambiguous}.
    \item If $\exists a_i \in \AA', a_j \notin \AA'$ with $a_i = a_j$, or symmetrically for outputs, the transaction is \textit{ambiguous}.
    \item If no connectable pair exists for any non-trivial subset of $\AA$ or $\BB$, the transaction is \textit{simple}.
\end{enumerate}
\end{theorem}

\begin{proof}
The first condition satisfies Lemma~1 of \cite{Yanovich2016b} (two distinct input subsets funding the same output set). The second condition satisfies Heuristic~3, which is a special case of Lemma~2 of \cite{Yanovich2016b}. The third condition follows directly from Definition~\ref{def:txType}: a transaction with no non-trivial acceptable partitions is simple.
\end{proof}

Unlike the preceding heuristics, which operate on structural templates, Heuristic~4 is a meta-heuristic that interprets the (in)completeness of solver results. It does not generate new partitions but rather deduces transaction type from the presence or absence of symmetries in the solver's output.

\section{Experimental pipeline}
\label{sec::pipepline}

\begin{figure}[h!]
    \includegraphics[width=0.5\linewidth]{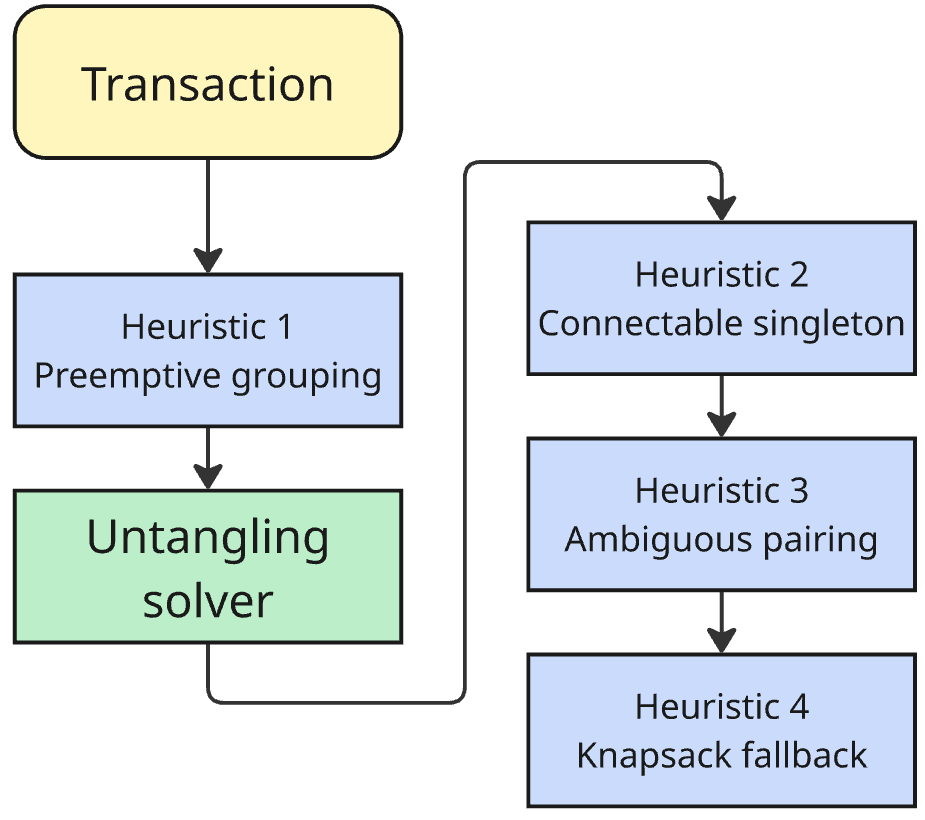}
    \centering
    \caption{Pipeline of heuristic application for transaction untangling}
    \label{figure:Pipeline}
\end{figure}

This section details the heuristic application pipeline, with execution order visualized in Figure~\ref{figure:Pipeline}.

Heuristic 1 is applied first due to its computational efficiency and composability with both the base untangling algorithm and all subsequent heuristics. After this heuristic is applied, the untangling solver is used to classify transactions.

Heuristics 2 and 3 both require iterative subset enumeration to identify connectable transaction pairs. Although optimized to exclude trivially small or excessively large subsets, these operations remain computationally intensive. Heuristic 3 further compounds this cost by requiring multiple iteration cycles, justifying their sequential execution after Heuristic 2.

Heuristic 4 incorporates a knapsack-solving subroutine whose complexity dominates runtime performance, as quantified in Table~\ref{tab:Heuristic times}. Consequently, it is reserved for terminal execution in the pipeline to avoid unnecessary computation on cases resolved by earlier heuristics.

\section{Numerical experiments}
\label{sec::Experiments}

Numerical experiments were conducted on the complete historical Bitcoin blockchain dataset, spanning from genesis block to block 882,421 (timestamped February 6, 2025). 

Transaction classification employed the untangling algorithm described in \cite{Larionov2023}, configured with a 1-second time limit per transaction. Our heuristic implementation utilized Python 3.7 with the optimized knapsack solver from package \cite{Knapsack}. All computations were executed on a workstation equipped with an AMD Ryzen 7 6800H processor (4.8 GHz base clock, 16 GB RAM). The sample source code of the workflow, is available on Github \cite{Englekh2026}.

\begin{figure}[h!]
    \includegraphics[width=0.6\linewidth]{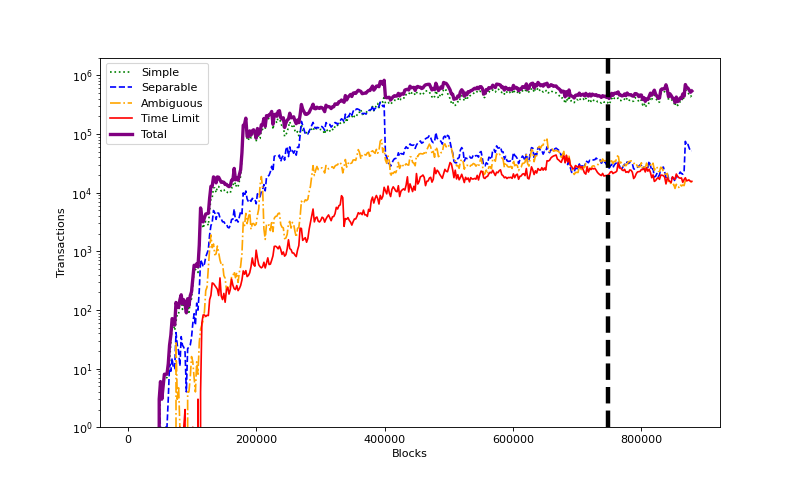}
    \centering
    \caption{Initial transaction classification results on the full dataset.}
    \label{figure:total}
\end{figure}

Figure~\ref{figure:total} shows the results of initial untangling of the Bitcoin dataset. The vertical line indicates the point at which results were reported in previous research \cite{Larionov2023}. Our results support the previous findings, showing a general trend of decreasing separable transactions over time, although a recent increase in their share is observable alongside ambiguous and time limit transactions.

This configuration yielded 5,023,471 transactions exceeding the 1-second timeout threshold. From this corpus, we selected a representative random subset of 50,000 transactions (approximately 1\% of all timeout-exceeding transactions) for detailed heuristic evaluation.

\begin{figure}[h!]
    \includegraphics[width=0.6\linewidth]{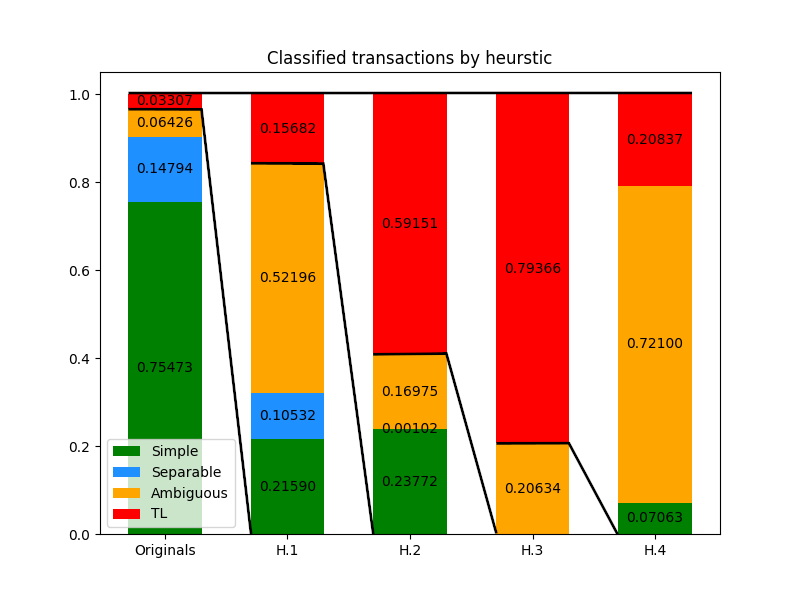}
    \centering
    \caption{Classification results of timeout transactions after heuristic pipeline.}
    \label{figure:Bars}
\end{figure}

Figure~\ref{figure:Bars} presents the classification distribution of timeout-exceeding transactions after applying the pipeline. Heuristic~1 resolves 84\% of initially ambiguous transactions. Heuristic~2 resolves 40\% of the remaining cases, while Heuristic~3 addresses 20\% of the subsequent remainder. Finally, Heuristic~4 resolves 80\% of the persistently ambiguous transactions.

Among the heuristics, only Heuristic 1 produced a considerable fraction of separable transactions (10.5\%). As anticipated for complex transactions with multiple inputs and outputs, most classified transactions were ambiguous, with each heuristic producing at least 16.9\% of the transactions (Heuristic 2). Heuristic 3 only found ambiguous transactions, as intended, while the others also identified simple transactions.

In the end, all heuristics together classified transactions as follows: the majority (62\%) were \textit{ambiguous}, 26\% were \textit{simple}, and 11\% were \textit{separable}--indicating uniquely resolvable transaction structures. This distribution contrasts sharply with the full dataset \cite{Larionov2023}, where 70\% are simple, 19\% separable, and only 9\% ambiguous, confirming that timeout-exceeding cases disproportionately contain structurally complex transactions.

\begin{table}[h]
    \centering
    \caption{Average heuristic computation time per transaction.}
    \label{tab:Heuristic times}
    \begin{tabular}{lc}
        \hline
        Heuristic & Avg. time (s) \\
        \hline
        1. Preemptive grouping & 0.01 \\
        2. Connectable singleton & 2.18 \\
        3. Ambiguous pairing & 45.91 \\
        4. Knapsack fallback & 85.95 \\
        \hline
    \end{tabular}
\end{table}

Table~\ref{tab:Heuristic times} validates our pipeline ordering through execution timing measurements. The data demonstrates strictly increasing computational complexity across the heuristic sequence, justifying the placement of resource-intensive operations later in the pipeline to maximize early-stage pruning efficiency.

The effectiveness of this design is underscored by comparison with brute-force approaches. In prior work \cite{Larionov2023}, increasing the computational budget for the base untangling algorithm from 1 to 300 seconds resolved only 25\% of previously timeout transactions, leaving 75\% still unclassified. By contrast, our heuristic pipeline resolves 98.5\% of timeout transactions with an effective average computation time of just 11.4 seconds per transaction--more than 25 times faster and achieving a 50-fold reduction in the unclassified rate. This improvement stems from early exit: 84\% of transactions are resolved by Heuristic~1 in 0.01 seconds, and only 7.7\% ever reach the expensive knapsack stage. The results demonstrate that structural heuristics, rather than raw computational power, are the key to untangling complex SSM transactions.

To assess the statistical reliability of our experimental results, we compute confidence intervals for the proportion of unresolved transactions. With \(n = 50\,000\) randomly sampled timeout transactions and an observed unresolved rate of \(1.53\%\), the standard error for this Bernoulli proportion is \(\sqrt{0.0153 \times 0.9847 / 50000} \approx 0.000549\). This yields a 95\% confidence interval of \(1.53\% \pm 1.96 \times 0.0549\% = [1.42\%,\; 1.64\%]\). The narrow interval width (\(0.22\%\)) demonstrates that our sample provides precise population estimates, confirming that the \(98.47\%\) resolution rate achieved by our heuristics is statistically robust and generalizable to the full set of timeout transactions.

We also validated the sampling strategy by benchmarking distributional statistics between the random sample and the full transaction set. Central tendency metrics showed near-perfect concordance: input medians were 37 (sample) versus 36 (full set), output medians were identical at 5 for both groups, and the median sum of inputs and outputs was exactly 71 in both datasets. Maximum values, while lower in the sample (inputs: 1,459 vs. 4,925; outputs: 3,001 vs. 10,001; sum: 3,006 vs. 10,004), represent the expected attenuation of extreme values inherent to random subsampling.  Given the robust agreement in medians--the primary indicator of typical transaction structure--we conclude the sample is representative of the full population for analytical purposes.

\section{Conclusion}
\label{sec:conclusion}

Prior work established a practical algorithm for SSM untangling, successfully classifying 98.6\% of all SSM transactions and leaving only 1.4\% unresolved due to computational time limits \cite{Larionov2023}. This paper complements that foundation by introducing a suite of heuristics specifically designed to address these hard residual cases. Applied to timeout transactions, our approach resolves 98.5\% of them, reducing the global unclassified SSM transaction rate from 1.4\% to 0.021\% of all SSM transactions--a 67-fold reduction.

Our analysis reveals a critical structural divergence: while non-timeout transactions are predominantly \textit{simple} (70\%) with minimal ambiguity (9\%) \cite{Larionov2023}, timeout cases exhibit inverse characteristics--62\% ambiguous versus 26\% simple. This confirms that computational complexity arises precisely where transactions permit multiple valid untanglings, validating our focus on ambiguous cases.

Though computationally intensive, the proposed heuristics provide essential instrumentation for resolving transactions intractable to brute-force methods. Notably, Heuristic~1 offers standalone value: its low resource requirements (0.01 s average runtime) enable integration into any transaction processing pipeline for complexity reduction without significant overhead.

These results establish that strategically designed heuristics can overcome computational barriers in 98.5\% of complex untangling cases. Nevertheless, several important directions remain for future investigation.

First, the residual 1.5\% of timeout transactions--those that resist all four heuristics even after pipeline application--warrant dedicated study. Manual inspection suggests these cases often involve intricate interactions between multiple structural conditions or near-equality constraints that narrowly evade our knapsack-based detection. Developing specialized techniques for this hard core, perhaps combining machine learning for pattern recognition with more sophisticated combinatorial optimization, could further reduce the unclassified rate toward zero.

Second, the applicability of our heuristic framework extends naturally beyond Bitcoin to the broader family of UTXO-based cryptocurrencies. Bitcoin Cash and Litecoin, as direct Bitcoin derivatives, share identical transaction semantics and thus can leverage our implementation without modification. More ambitiously, adapting these heuristics to Cardano's Extended UTXO (EUTXO) model presents both challenges and opportunities. Our prior work on Cardano \cite{Chegenizadeh2025} established that SSM untangling in the multi-asset, staking-enabled EUTXO setting is NP-complete and formulated a vectorized generalization of the original algorithm. Integrating the heuristic principles introduced here--preemptive grouping, singleton reduction, ambiguity detection via value collisions, and knapsack fallback--into the EUTXO framework could yield similar efficiency gains for Cardano's most complex transactions. This cross-chain generalization would represent a significant step toward a unified transaction analysis toolkit for all UTXO-based blockchains.

Third, the relationship between untangling and address clustering heuristics merits deeper exploration. Prior work \cite{Larionov2024} demonstrated that applying Common Spending and One-Time Change heuristics to untangled subtransactions can reveal previously hidden address linkages. Our heuristics, by resolving vastly more timeout transactions, unlock a substantially larger corpus of separable transactions for such analysis. A systematic study of how untangling-induced clustering affects downstream forensic applications--including entity identification, money flow tracing, and risk scoring--could quantify the practical investigative value of the resolution rates we report.

Finally, the ambiguous transactions that constitute the majority of our resolved timeout cases present a more subtle opportunity. While ambiguity precludes unique untangling, it does not render a transaction entirely uninformative. The set of all valid minimal partitions defines a distribution of possible flow decompositions. Analyzing this distribution--for instance, identifying which input-output pairings are present in all valid partitions versus those that vary--can yield probabilistic flow information even when deterministic reconstruction is impossible. Developing methods to extract such partial structural insights from ambiguous transactions and integrate them into forensic workflows remains an open and promising research direction.

Collectively, these directions point toward a comprehensive framework for UTXO transaction analysis that combines algorithmic untangling, heuristic acceleration, cross-chain generalization, and probabilistic flow inference. The 67-fold reduction in unclassified transactions achieved in this paper demonstrates that even NP-hard problems can yield to well-crafted heuristics when guided by empirical observation of real-world transaction patterns.

\bibliographystyle{splncs04}
\bibliography{refs2025}

\end{document}